\documentclass[12pt,letterpaper]{amsart}
\usepackage{cite}
\usepackage{amsmath,amssymb,amsthm}
\usepackage{algorithm}
\usepackage{algpseudocode}
\usepackage{enumerate}
\usepackage{enumitem}
\usepackage{pgf,tikz}
\usepackage{thmtools}
\usepackage{lscape}
\usepackage{subcaption}
\usepackage[colorlinks]{hyperref}
\usepackage{mathrsfs}
\usepackage{multicol}
\usepackage{graphicx}
\usepackage{booktabs}

\usepackage[all]{xy}
\usepackage{mathtools}
\usepackage[utf8]{inputenc}
\usepackage[T1]{fontenc}
\usepackage{lmodern}
\usepackage{exscale,relsize}
\usepackage{xcolor}

\usepackage[pagewise]{lineno}

\usetikzlibrary{arrows,matrix}
\usetikzlibrary{shapes.geometric}

\newtheorem{theorem}{Theorem}[section]
\newtheorem{lemma}[theorem]{Lemma} 
 
\newcommand{\soc}{\mathrm{soc}}

\newtheorem{teo}{Theorem}[section]
\newtheorem{remark}[teo]{Remark}
\newtheorem{cor}[teo]{Corollary}

\title[On reversible and reversible-complementary DNA codes]{ On reversible and reversible-comple-\\mentary DNA codes over $\mathbb{F}_{4}$}

\email{eliasjaviergarcia@gmail.com}
\date{\today}

\begin{document}

\maketitle

\begin{center}
  
  {\normalsize  E. J. Garc\'ia-Claro}\footnote{ORCID: \protect\href{https://orcid.org/0000-0003-4047-8554}{0000-0003-4047-8554}\par}\\
  {\small Departamento de Matem\'aticas\\
   Universidad Aut\'onoma Metropolitana-Iztapalapa\\
 Ciudad de M\'exico, M\'exico\par}
\end{center}

\begin{abstract}

A method to construct and count all the linear codes (of arbitrary length) in $\mathbb{F}_{4}$ that are invariant under reverse permutation and that contain the repetition code is presented. These codes are suitable for constructing DNA codes that satisfy the reverse and reverse-complement constraints.  By analyzing a module-theoretic structure of these codes, their generating matrices are characterized in terms of their isomorphism type, and explicit formulas for counting them are provided. The proposed construction method based on this characterization outperforms the one given by Abualrub  et al. \cite{abualrub} for cyclic codes (of odd length) over $\mathbb{F}_{4}$, and the counting method solves a problem that can not be solved using the one given by Fripertinger \cite{fripertinger} for invariant subspaces under a linear endomorphism of $\mathbb{F}_{q}^{n}$. Additionally, several upper bounds and an identity for the minimum Hamming distance of certain reversible codes are provided.

\end{abstract}

\textit{keywords}: DNA codes; reversible DNA codes, reverse codes.


\section*{Introduction}

Deoxyribonucleic acid (DNA) is a polymer that contains instructions for the biological development of living organisms and viruses. This polymer is formed by two stands linked in a twisted double helix shape. Each of these  strands is a sequence of four possible nucleotides, adenine ($A$), guanine ($G$), cytosine ($C$) and thymine ($T$) with ends chemically polar called $3'$ and $5'$, e.g., $3'-ATAGGCTC-5'$. Thus, any DNA strand can be consider as a sequence of letters in the alphabet $\mathcal{A}:=\{A, C, T, G\}$ (i.e., an element of $\mathcal{A}^{n}$ for some $n\in \mathbb{Z}_{>0}$). The Watson-Crick complement (WCC) of a DNA strand $x$ is the sequence $c(x)$ or $x^{c}$ obtained by replacing each $A$ by $T$ and vice versa and each $C$ by $G$ and viceversa (while switching the $3'$ and $5'$ ends). For example, the WCC of $3'-ATAGGCTC-5'$ is $5'-TATCCGAG-3'$ (in the context of coding theory the ends are usually omitted and so the WCC of $ATAGGCTC$ is $(ATAGGCTC)^{c}:=TATCCGAG$). A code of length $n$ over $\mathcal{A}$ is a  subset of $\mathcal{A}^{n}$.\\
DNA Hybridization is the process in which two DNA strands bond together to form a double stranded DNA molecule. This process follows the base pair rule in which the $A$ and $T$ are bonded together and $C$ and $G$ are bonded together. Hybridization is necessary for DNA synthesis, which is crucial in areas such as DNA computing or DNA digital data storage, among others. DNA computing combines genetic data analysis with computer sciences to solve  computationally difficult problems. This field started in 1994 when Adleman solved a hard (NP-complete) in graph theory \cite{adleman}, namely ``the directed Hamiltonian path problem'' (on a graph with seven nodes). See \cite{adleman2, lipton} for more outstanding applications of DNA computing. On the other hand, DNA digital data storage uses the chemical structure of DNA nucleotides as a tool to store digital information. This presents a strong technology for large-scale digital record, which has great potential for multi-millennia-long archiving tasks due to DNA resistance to degradation over time, provided adequate conditions  \cite{masud, dna-storage}. For instance, it is known that a 1g of DNA contains about $4.2\times 10^{21}$  bits of information ($\approx100$ TB) \cite{abualrub}, and that a DNA sample can endure up to two million years \cite{kjaer}.\\
Ideally, it is expected that a single DNA strand bonds to its WCC, but this process is commonly affected by the formation of diverse undesired bonds such as the the one between a strand with the WCC of another strand (e.g., if the strand $5'-GGAC\textcolor{red}{T}GACG-3'$ links to $3'-CCTG\textcolor{red}{G}CTGC-5'$, causing a defect called DNA mismatch) or with itself (forming a structure called stem-loop). Thus a problem of interest for DNA synthesis is the one of constructing codes (i.e., collection of strands of a fixed length) that are suitable  to avoid undesired bonds that might occur during the hybridization process. In order to do so, several constraints have been studied such that the codes  satisfying them are less likely  to form these undesired pairing \cite{king, marathe, frutos}.\\
If $\mathfrak{A}\neq \emptyset$, the reverse mapping in $\mathfrak{A}^{n}$ is the mapping $r:\mathfrak{A}^{n}\rightarrow \mathfrak{A}^{n}$ given by 
$r(a_{1},a_{2},...,a_{n-1},a_{n}):=(a_{n},a_{n-1},...,a_{2},a_{1})$. For $x\in \mathfrak{A}^{n}$, the notation $r(x)$ or $x^{r}$ are adopted, indistinctly. Let $\mathcal{C}\subseteq \mathcal{A}^{n}$ be a code
\begin{enumerate}
\item If $d\leq d(x^{r}, y)$ $\forall x,y\in \mathcal{C}$ with $x\neq y$, it is said that $\mathcal{C}$ satisfies the reverse constraint (for $d$) or that $\mathcal{C}$ is a reverse code \cite{king}.
\item If $d\leq d(x^{r}, y^{c})$ $\forall x,y\in \mathcal{C}$ with $x\neq y$, it is said that $\mathcal{C}$ satisfies the reverse-complement constraint (for $d$) or that $\mathcal{C}$ is a reverse-complement code (for $d$) \cite{king}.
\end{enumerate}

Among the codes satisfying constraint $1$ are the ones such that $r(\mathcal{C})=\mathcal{C}$ (called reversible codes \cite{massey}); and, among these last, are those satisfying that $c(\mathcal{C})=\mathcal{C}$ (here called reversible-complementary codes), that satisfy constraints $1-2$. 
  Many authors have aboard the study of DNA codes over $\mathbb{F}_{4}$\cite{abualrub,gaborit, smith, aboluion}. In particular, Gaborit and King characterize reverse codes (not necessarily linear) as those with a permutation group of even order \cite[Lemma 1]{gaborit} and offer several lower bounds for the maximum size of a code in the family of reverse and reverse-complement codes. Later in \cite{abualrub}, Abualrub, Ghrayeb, and Zeng developed a theory for constructing linear and additive cyclic codes of odd length. In the linear case, if  $n=7,11,13$, their approach only produces the repetition code, and two codes distinct from the repetition code, for $n=9$. \\ 
This work will be focused on developing algebraic techniques for constructing and counting all the reversible and the reversible-complement\\-ary codes. The manuscript is organized as follows. In Section \ref{s1}, a characterization of the studied codes is given in terms of their generating matrices; upper bounds for the minimum Hamming distance of reversible codes, that can be better than the Singleton bound, are introduced; and an efficient method for constructing reversible and reversible-complementary codes is presented. Later, in Section \ref{s2}, some formulas to count all the reversible and reversible-complementary linear codes, depending on their isomorphism type,  are presented.

\section{Computation of  reversible and reversible-complementary codes over $\mathbb{F}_{4}$}\label{s1}

In this section, a method to compute all the submodules of $V$ in an efficient way is presented, i.e., in a way that no module is obtained more than once as a sum of indecomposable modules .\\
From now on, $F:=\mathbb{F}_{4}=\{0, 1, \alpha, \alpha^{2}\}$ (where $\alpha^{2}=\alpha + 1$) is the finite field of order $4$, and $V:=F^n$ with $n\geq 2$. $r: V\to V$ is the reverse permutation automorphism, i.e., the linear automorphism of $V$ given by $r(a_1,a_2,\ldots, a_{n-1},a_n):=(a_n, a_{n-1}, \ldots,a_2, a_1)$. \\
If $\mathcal{A}=\{A, C, T, G\}$ and $\varphi: \mathcal{A} \rightarrow \mathbb{F}_{4}$ is the mapping given by $\varphi(A)=0$,     $\varphi(C)=\alpha$, $\varphi(T)=1$, and $\varphi(G)=\alpha^{2}$, then $\varphi^{n}: \mathcal{A}^{n}\rightarrow V$ defined as $\varphi^{n}((a_{i})_{i=1}^{n})=(\varphi(a_{i}))_{i=1}^{n}$ is a bijection such that $\varphi^{n}((a_{i}^{c})_{i=1}^{n})= (\varphi(a_{i}))_{i=1}^{n} + (1,1,...,1)$. Thus if $\mathcal{C}\subseteq \mathcal{A}^{n}$, $r(\mathcal{C})=\mathcal{C}$ if and only if  $r(\varphi^{n}(\mathcal{C}))=\varphi^{n}(\mathcal{C})$, and $c(\mathcal{C})=\mathcal{C}$ if and only if $\varphi^{n}(\mathcal{C})+(1,1,...,1)=\varphi^{n}(\mathcal{C})$.  Hence  $\varphi^{n}$ gives a one-to-one correspondence between reversible (reversible-complementary) codes in $\mathcal{A}^{n}$ and reversible codes (reversible codes containing the repetition code) in $V$. The existence of this bijection motivates the construction and counting of reversible (reversible-complementary) codes in $V$.\\
Let $R=\frac{F[x]}{\langle x^{2}+1\rangle}$. Since $V$ has a left module structure defined by $ (ax+b) \cdot v = ar(v)+bv $  for any $ax+b\in R$ and $v \in V$, a linear code $\mathcal{C}\subseteq V$  is reversible  if and only if $\mathcal{C}$  is a submodule of $V$. From now on, every module will be a left module. If $M$ is a module, and $t$ is a non-negative integer, the external direct sum $M \oplus \cdots \oplus M$ ($t$ times) is denoted by $tM$, where $0M=0$. It is an easy exercise to verify that any submodule (reversible code) of $V$ is a direct sum of copies of $R$  and/or $F$ (since $R$ and $F$ are the only indecomposable modules, this coincides with Krull-Schmidt Theorem).\\
From now on, $T:=r+\mathrm{id}_V$, $K:=ker(T)$, $I:=im(T)$, and $\mathbf{1}$ will denote the repetition code. If $C\subseteq V$ is a linear code, $C^{\perp}:=\{v\in V: v\cdot c=0 \, \text{ for all } c\in C\}$  is the dual of $C$ under the standard dot product.

\begin{lemma}\label{lem2.2}
Let $\beta=\{ e_1, \ldots, e_n \}$ be the canonical ordered basis of $V$. The following statements hold:
\begin{enumerate}
\item $K$ is the set of fixed points under $r$, and $\gamma:=\{ e_{i}+r(e_{i})\, : \, i=1,\ldots, \lfloor n/2 \rfloor  \}$ is a basis for $I$. In particular, $\dim(I) = \lfloor n/2 \rfloor$ and $\dim (K)=\lceil n/2 \rceil$.
\item Any linear code contained in $I$ is self-orthogonal.
\item $K =  I$ if $n$ is even; and
$K=I\oplus \mathbf{1}$ if $n$ is odd. In particular, $K=I^{\perp}$. 
\item  $V \cong_R \lfloor n/2\rfloor R$ if $n$ is even; and $V \cong_R (\lfloor n/2\rfloor R) \oplus F$ if  $n$  is odd.
\end{enumerate}
\end{lemma}
\begin{proof}

\begin{enumerate}

\item For $v\in V$, $T(v)=r(v)+v=0$ if and only if $r(v)=v$. Since $\gamma$ is a linear independent subset of $I$ and $I=T(V)=T(\langle \beta \rangle_{F})=\langle T(\beta) \rangle_{F}=\langle \gamma \rangle_{F}$,  $\gamma$ is a basis for $I$. The rest follows from the rank-nullity theorem and the fact that $n=\lceil n/2 \rceil+ \lfloor n/2 \rfloor$.
\item Let $C\subseteq I$ be a linear code and $c\in C$. Then, by part $1$,  $c=\left(\sum_{i=1}^{\lfloor n/2 \rfloor}c_{i}(e_{i}+r(e_{i}))\right)\cdot (e_{j}+r(e_{j}))=c_{j}(e_{j}+r(e_{j}))\cdot (e_{j}+r(e_{j}))=2c_{j}=0$ for $j=1,...,\lfloor n/2 \rfloor$ and so any element in $C$ is orthogonal to any element in $I$. Therefore $C\subseteq C^{\perp}$.
\item Since $T^{2}=0$, $I\subseteq K$. Thus, if $n$ is even, $I=K$ (by part $1$). In particular,  if $n$ is even,  $K=I\subseteq I^{\perp}$ (by part $2$), then $K=I^{\perp}$ (by part $1$). On the other hand, if $n$ is odd, $\mathbf{1}\subset K-I$ thus $K=I \oplus \mathbf{1}$. In addition $\gamma \cup \{(1,1,...,1)\}$ is a basis for $K$ whose elements are orthogonal to the ones in $\gamma$, so that $K\subseteq I^{\perp}$ and so $K= I^{\perp}$ (by part $1$).
\item If $1\leq i \leq\lfloor n/2 \rfloor$, $M_i:=\langle e_i, r(e_i) \rangle_F$ is a submodule of $V$ non-isomorphic to $F\oplus F$ (because $e_{i}$ is not a fixed point of $r$) and so $M_i\cong_R R$. In addition, if $n$ is odd, $M_{\lceil n/2 \rceil} \cong_R F$. The rest follows form the fact that $V = \bigoplus_{i=1}^{\lceil n/2 \rceil} M_i$.
\end{enumerate}
 
\end{proof}

Let $M\neq 0$ be a module. The socle $soc(M)$ is the sum of all simple submodules of $M$. If $soc(M)=M$ it is said that $M$ is a semisimple module. The standard notation $M\leq V$ will be used to denote that $M$ is a submodule of $V$. If $M,L\leq V$, it will be said that $M$ is a \textbf{socle-extension} of  $L$ or that $L$ \textbf{extends} to $M$ if $\soc(M)=L$. Since $soc(V)=K$, any submodule $0\neq M$ of $V$ is a socle-extension of a subspace of $K$. Thus we will focus on computing and counting the submodules that are socle-extensions of subspaces of $K$. 

\begin{lemma}\label{lem2.3}
Let $0\neq M\leq V$ and $S\leq V$ be simple. Then the following holds:
\begin{enumerate}
\item $\soc (M)=K\cap M$. In particular, $M$ is semisimple if and only if $M\subseteq K$.

\item $M$ is simple if and only if $M=\langle s \rangle_F$ for some $s\in K - \{0\}$.

\item $M\cong_R R$ if and only if there exists $ m \in M-soc(M)$ such that $M=\langle m,r(m) \rangle_F$.

\item There exists a submodule $M\cong_R R$ of $V$ such that $S$ extends to $M$ if and only if $S \subseteq I$.


\end{enumerate} 
\end{lemma}

\begin{proof}
\begin{enumerate}

\item Since $R$ is local with maximal ideal $J=R(x+1)$, $\soc(M)=\{ m \in M: \ J \cdot m=0\}=\{ m \in M: \ r(m)+m=0\}=K\cap M$. The rest is clear.

\item It follows from Lemma \ref{lem2.2}, part $1$.

\item Suppose $M\cong_R R$. Let $\varphi:R\rightarrow M$ be an isomorphism of modules, then  $\varphi(R)=  \varphi(\langle 1,x \rangle_F) = \langle \varphi(1), \varphi(x) \rangle_F= \langle \varphi(1), x\varphi(1) \rangle_F$ $=\langle \varphi(1),$ $ r(\varphi(1)) \rangle_F= \langle m, r(m) \rangle_F  $ where $m=\varphi(1) \in M-soc(M)$ (because $1\in R-soc(R)$). Conversely, suppose there exists $ m \in M-soc(M)$ such that $M=\langle m,r(m) \rangle_F$. Since the only eigenvalue of $r$ is $1$ and $m\in M-soc(M)$, then $r(m)\neq \lambda m$ for all $\lambda\in F$. Hence $M$ has dimension $2$ and $M\not\cong F\oplus F$. Therefore $M\cong_{R} R$.

\item Let $M\leq V$ with $M\cong_R R$ such that $soc(M)=S$. If $\varphi: R\rightarrow M$ is an isomorphism of modules, then $0\neq \varphi (x+1)=\varphi(x)+\varphi(1)=x\varphi(1)+\varphi(1)=r(\varphi(1))+\varphi(1)=T(\varphi(1))\in I\cap M$. Thus, since $I\subseteq K$ (by Lemma \ref{lem2.2}, part $3$), then $0\neq I\cap M\subseteq K\cap M=soc(M)=S$ (by part $1$), and so $S=I\cap M\subset I$. Converseley, if $S \subseteq I$, then $K\subsetneq T^{-1}(S)\neq \emptyset$. Thus, if $v\in T^{-1}(S)-K$, $T(v)=r(v)+v\in S-\{0\}$ and so $S\subseteq M:= \left\langle v, r(v) \right\rangle_{F}$ which is isomorphic to $R$ (by part $3$). Hence $S$ extends to $M$.
\end{enumerate}
\end{proof}



Since, $soc(M)\subseteq  soc(V)=K$, $\dim(soc(M))=t+s\leq \lceil n/2 \rceil$ (by  Lemma \ref{lem2.2}, part $1$). In addition, $t\leq \lfloor n/ 2 \rfloor$ (by Lemma \ref{lem2.2}, part $4$). Thus if $M\leq V$, then $M\cong_{R}tR \oplus sF$  for some $t\in \{0,..., \lfloor n/2 \rfloor\}$  and $s\in \{0,..., \lceil n/2 \rceil-t \}$.\\
Let $0\neq M$ be a module. If $X=\{M_{i}\}_{i=1}^{t}$ is a collection of submodules of $M$, such that $\sum_{i=1}^{t}M_{i}=\oplus_{i=1}^{t}M_{i}$, then it is said that $X$ is an independent set of submodules of $M$.

\begin{remark}\label{ind}
Let $\beta=\{ M_i \}_{i=1}^{k} $ be a collection of submodules of $V$ isomorphic to $R$ and $\beta'= \{soc(M_{i})\}_{i=1}^{k}$. Then $\beta$ is independent if and only if $\beta'$ is independent.
\end{remark}

 Let $t\in \{1,...,\lfloor n/2 \rfloor \}$ and $\{S_{i}\}_{i=1}^{t}$ be an independent set of simple submodules of $I$. Let $\overline{\phantom{R}}: V \rightarrow V/\oplus_{i=1}^{t}S_{i}$ be given by $u\mapsto \overline{u}:= u + \oplus_{i=1}^{t}S_{i}$ for all $u\in V$. Then the relation $\sim$ defined in $\prod_{i=1}^{t}(T^{-1}(S_{i})-K)$ as $(u_{i})_{i=1}^{t}\sim (w_{i})_{i=1}^{t}$ if and only if $\langle\{ \overline{u_{i}}\}_{i=1}^{t}\rangle_{F}=\langle \{ \overline{w_{i}} \}_{i=1}^{t}\rangle_{F}$ is an equivalence relation.


\begin{teo}\label{const-tRsF}
Let $M\leq V$, $t\in \{1,...,\lfloor n/2 \rfloor\}$ and $s\in \{0,..., \lceil n/2 \rceil-t \}$. Then $M\cong_{R} tR\oplus sF$ if and only if there exists an independent set $\{S_{i}\}_{i=1}^{t}$  of simple submodules of $I$ and $(u_{i})_{i=1}^{t}\in \prod_{i=1}^{t}(T^{-1}(S_{i})-K)$ such that $M=\left(\oplus_{i=1}^{t}\langle u_{i}, r(u_{i})\rangle_{F}\right)\oplus L$, where $L$ is a subspace of dimension $s$ of $K$ such that $L\cap\left(\oplus_{i=1}^{t}S_{i}\right)=0$. In addition, if $(u_{i})_{i=1}^{t}, (w_{i})_{i=1}^{t}\in \prod_{i=1}^{t}(T^{-1}(S_{i})-K)$, then  $\sum_{i=1}^{t}\langle u_{i}, r(u_{i})\rangle_{F}=\sum_{i=1}^{t}\langle  w_{i}, r(w_{i})$ $\rangle_{F}$ if and only if $(u_{i})_{i=1}^{t}\sim (w_{i})_{i=1}^{t}$.

\end{teo}

\begin{proof}
Let $N\leq V$. It will be shown that $N\cong_{R} tR$ if and only if there exists an independent set $\{S_{i}\}_{i=1}^{t}$  of simple submodules of $I$ and $(u_{i})_{i=1}^{t}\in \prod_{i=1}^{t}(T^{-1}(S_{i})-K)$ such that $N=\oplus_{i=1}^{t}\langle u_{i}, r(u_{i})\rangle_{F}$. Suppose that $N\cong_{R} tR$, then there exists a set $\{N_{i}\}_{i=1}^{t}$ of submodules of $V$ such that $N_{i}\cong_{R} R$ for all $i$, and $N=\oplus_{i=1}^{t}N_{i}$. Let $S_{i}:=soc(N_{i})$ for all $i$, then   $\{S_{i}\}_{i=1}^{t}$ is an independent set of simple submodules of $I$ (by Remark \ref{ind} and Lemma \ref{lem2.3}, part $4$). Thus, by Lemma \ref{lem2.3} (part $3$),  there exists $u_{i}\in N_{i}-S_{i}\subseteq T^{-1}(S_{i})-K$ such that $N_{i}=\langle u_{i}, r(u_{i})\rangle_{F}$ for all $i$. Conversely, suppose that there exists an independent set $\{S_{i}\}_{i=1}^{t}$  of simple submodules of $I$ and $(u_{i})_{i=1}^{t}\in \prod_{i=1}^{t}(T^{-1}(S_{i})-K)$ such that $N=\oplus_{i=1}^{t}N_{i}$ with $N_{i}:=\langle u_{i}, r(u_{i})\rangle_{F}$ for all $i$. Then $S_{i}=soc(N_{i})$ and $u_{i}\in N_{i}-S_{i}$ so that  $N_{i}\cong_{R}R$ for all $i$ (by Lemma \ref{lem2.3}, part $3$) and so $N\cong_{R}tR$ (by Remark \ref{ind}). 
  Then, the first equivalence follows from Lemma \ref{lem2.2} and the fact that a submodule of $V$ intersects another submodule trivially if and only if it intersects its socle trivially. In addition, if $(u_{i})_{i=1}^{t}, (w_{i})_{i=1}^{t}\in \prod_{i=1}^{t}(T^{-1}(S_{i})-K)$, since $\overline{\langle u_{i}, r(u_{i})\rangle}_{F}=\overline{\langle u_{i}, T(u_{i})\rangle}_{F}=\langle \overline{u_{i}}\rangle_{F}$ and  $\overline{\langle w_{i}, r(w_{i})\rangle}_{F}=\langle \overline{w_{i}}\rangle_{F}$ for all $i$, then $\sum_{i=1}^{t}\langle u_{i}, r(u_{i})\rangle_{F}=\sum_{i=1}^{t}\langle  w_{i}, r(w_{i})\rangle_{F}$ if and only if $\langle\{ \overline{u_{i}}\}_{i=1}^{t}\rangle_{F}=\overline{\sum_{i=1}^{t}\langle u_{i}, r(u_{i})\rangle}_{F}=\overline{\sum_{i=1}^{t}\langle  w_{i}, r(w_{i})\rangle}_{F}=\langle \{ \overline{w_{i}} \}_{i=1}^{t}\rangle_{F}$; or equivalently, $(u_{i})_{i=1}^{t}\sim (w_{i})_{i=1}^{t}$.
\end{proof}

Let $\widehat{}:V\longrightarrow V$ be the mapping given by $v=(v_{1},...,v_{n})\mapsto \widehat{v}:=(v_{1},...,v_{\lfloor n/2 \rfloor},0,...,0)$.

\begin{cor}\label{cor-const-tRsF}
Let $t\in \{1,...,\lfloor n/2 \rfloor\}$, $s\in \{0,..., \lceil n/2 \rceil-t \}$ and $C\leq V$. Then the following holds:

\begin{enumerate}

\item   $C\cong_{R}tR$ (with $(1,1,...,1)\in C$) if and only if $C$ has a generator matrix  of the form 
\[
G := \begin{bmatrix}
\widehat{s_{1}} + k_{1} \\
r(\widehat{s_{1}}) + k_{1} \\
\vdots \\
\widehat{s_{t}} + k_{t} \\
r(\widehat{s_{t}}) + k_{t}\\
\end{bmatrix}_{2t \times n}\;\;\; (\text{with } s_{1}=(1,1,...,1))
\]
where $\{s_{i}\}_{i=1}^{t}\subseteq I$ is a linear independent and  $\{k_{i}\}_{i=1}^{t}\subseteq K$.
 
\item $C\cong_{R}tR \oplus sF$ (with $(1,1,...,1)\in C$) if and only if $C$ has a generator matrix  of the form 

\[
G' := \left[\begin{array}{c}
G\hspace{0.1cm}\\
 \hline 
l_{1}\\
\vdots \\
l_{s}
\end{array}\right]_{(2t+s) \times n}\;\;\; (\text{with } s_{1} \text{ or } l_{1} \text{ equal to } (1,1,...,1))
\]
where $\{s_{i}\}_{i=1}^{t}\subseteq I$ and   $\{l_{i}\}_{i=1}^{t}\subseteq K-\langle \{s_{i}\}_{i=1}^{t} \rangle_{F}$ are linear independent sets, $\{k_{i}\}_{i=1}^{t}\subseteq K$, and $G$ is as in part $1$.

\end{enumerate}

\end{cor}

\begin{theorem}\label{upperbound}
Let $0\neq C\leq V$, then
\[d(C)\leq \begin{cases}2d(\widehat{soc(C)}) & \mbox{ if } soc(C)\subseteq I\\
                    2d(\widehat{soc(C)})+1 & \mbox{ if }  soc(C)\nsubseteq I
  \end{cases},
  \]
and if $d(\widehat{soc(C)})< \frac{1}{2}\lfloor \frac{n}{2} \rfloor$, these upper bounds for $d(C)$ are better than the Singleton bound. In particular, if $C\subseteq I$, then  $d(C)=2d(\widehat{C})$.
\end{theorem}

\begin{proof}

Note that, by Lemma \ref{lem2.2} (part $1$), any $u=(u_{1},...,u_{n})\in I$ is such that $u_{i}=u_{n-i+1}$ for $i=1,..,\lfloor n/2 \rfloor$. Let $v,w\in soc(C)$. Since $soc(C)\subseteq K$ (by  Lemma \ref{lem2.3}, part $1$) and $I\subseteq K$ (by Lemma \ref{lem2.2}, part $3$), then $v,w\in I$; or $v,w\in K-I$; or $v\in I$ and $w\in K-I$.  If $v,w\in I$, then $d(v,w)=2d(\widehat{v},\widehat{w})$ because $v_{i}=v_{n-i+1}$ and $w_{i}=w_{n-i+1}$ for $i=1,..,\lfloor n/2 \rfloor$ and the fact that $v$ and $w$ both have middle entry equal to $0$ when $n$ is odd. The other two cases  are only possible when $n$ is odd (by Lemma  \ref{lem2.2}, part $3$).  If $n$ is odd, $v,w\in K-I$ and  $Q\in V$ is the codeword having the middle entry equal to $1$ and zero elsewhere, then $v=v'+\lambda_{1}Q$ and $w=w'+\lambda_{2}Q$ where $v',w'\in I$ and $\lambda_{1},\lambda_{2}\in F-\{0\}$ (by Lemma  \ref{lem2.2}, part $3$). It is easy to check that $d(v,w)=d(v',w')+d(\lambda_{1}Q, \lambda_{2}Q)$ so that $d(v,w)=2d(\widehat{v},\widehat{w})$ if $\lambda_{1}=\lambda_{2}$, or $d(v,w)=2d(\widehat{v},\widehat{w})+1$ if $\lambda_{1}\neq \lambda_{2}$. Similarly, it can be shown that if $n$ is odd, $v\in K$  and $w\in K-I$, then $d(v,w)=2d(\widehat{v},\widehat{w})+1$. Hence for $v,w\in soc(C)$
\[d(v,w)\begin{cases}=2d(\widehat{v},\widehat{w}) & \mbox{ if } v,w\in I\\
                     \leq 2d(\widehat{v},\widehat{w})+1 & \mbox{ if } v,w\in K-I\\
                     = 2d(\widehat{v},\widehat{w})+1 & \mbox{ if }  v\in K \wedge w\in K-I\\                     
  \end{cases}.
  \]
Therefore $d(C)\leq d(soc(C))=2d(\widehat{soc(C)})$ if $ soc(C)\subseteq I$; and $d(C)\leq 2d(\widehat{soc(C)})+1$ if $soc(C)\nsubseteq I$. By Lemmas \ref{lem2.2} and \ref{lem2.3} (part $1$), $\dim(\widehat{soc(C)})\leq \lceil \frac{n}{2}\rceil$. Thus, if $d(\widehat{soc(C)})< \frac{1}{2}\lfloor \frac{n}{2} \rfloor$ these bounds for $d(C)$ are better than the Singleton bound. The rest follows from Lemma \ref{lem2.3} (part $1$).
\end{proof}

By Lemma \ref{lem2.3} (part $1$), computing the submodules of $V$ isomorphic to $sF$ is the same as computing the subspaces of $K$ of dimension $s$ for $0\leq s \leq \lceil n/2 \rceil $. Now, a method to compute all the submodules of $V$ isomorphic to $tR\oplus sF$ with $1\leq t \leq \lfloor n/2 \rfloor $ and $0\leq s \leq \lceil n/2 \rceil -t$ will be provided.

\begin{remark}[\textbf{Method to compute the submodules of $\mathbf{V}$ isomorphic to} $\mathbf{tR\oplus sF}$]\label{rmk-const-tR}
If $W$ is a subspace of $I$ with basis $\{s_{i}\}_{i=1}^{t}$, $S_{i}:=\langle s_{i}\rangle_{F}$ for all $i$, and  $\tau=\{(u_{i1})_{i=1}^{t},...,(u_{il})_{i=1}^{t}\}\subset \prod_{i=1}^{t}(T^{-1}(S_{i})-K)$ is a set of representatives of equivalence classes under $\sim$, the modules $N_{j}:=\oplus_{i=1}^{t}\langle u_{ij}, r(u_{ij})\rangle_{F}$ with $j=1,...,l$ are exactly the submodules of $V$ (isomorphic to $tR$) that are socle-extensions of $W$ (by Theorem \ref{const-tRsF}). Thus, to compute these modules $N_{j}'s$, one may start by constructing a module $N_{1}=\oplus_{i=1}^{t}\langle u_{i1}, r(u_{i1})\rangle_{F}$ taking an arbitrary element $(u_{i1})_{i=1}^{t}\in \prod_{i=1}^{t}(T^{-1}(S_{i})-K)=\prod_{i=1}^{t} \left( \sqcup_{\lambda \in F-\{0\} } \lambda \widehat{s}_{i} + K \right)$ (here $\sqcup$ denotes the disjoint union of sets). Later, if $k\geq 2$ and $\left(\prod_{i=1}^{t}(T^{-1}(S_{i})-K\right.$ $\left.)\right)-\cup_{j=1}^{k-1}N_{j}\neq \emptyset$, continue constructing $N_{k}=\oplus_{i=1}^{t}\langle u_{ik}, r(u_{ik})\rangle_{F}\cong_{R}tR$ with $(u_{ik})_{i=1}^{t}\in \left(\prod_{i=1}^{t}(T^{-1}(S_{i})\right. $ $ \left.-K)\right)-\cup_{j=1}^{k-1}N_{j}$. This process should be repeated for every subspace $W$ of $I$ of dimension $t$ to obtain all the submodules of $V$ isomorphic to $tR$.   
Now, if $N\leq V$ is such that $N\cong_{R}tR$ one may take $J_{1}\leq K$ of dimension $s$ such that $J_{1}\cap soc(N)=0$ and construct the module $N+ J_{1}$. In addition, if $k\geq 2$, and $J_{k}\leq K$  is a subspace dimension $s$ such that $J_{k}\cap soc(N)=0$ and $J_{k}\nsubseteq N+J_{u}$ for $u=1,..., k-1$, continue the process constructing $N+J_{k}\cong_{R}tR\oplus sF$.  This process should be repeated for every submodule $N$ isomorphic to $tR$ to obtain all the submodules of $V$ isomorphic to $tR\oplus sF$. By Corollary \ref{cor-const-tRsF} (part $2$) one can construct similarly all the submodules of $V$ isomorphic to $tR\oplus sF$ containing $\mathbf{1}$.
\end{remark}

\section{Counting reversible and  reversible-complementary codes over $\mathbb{F}_{4}$} \label{s2}

This section presents different formulas for counting  submodules of $V$ based on their isomorphism type.


 \begin{lemma}\label{lem3.3} 
 Let $v \in V$ such that $S:=\langle T(v) \rangle_F$ is a simple module. Let $0\neq M\leq V$ and $\mathcal{N}=\{N \leq M: N\cong_R R \}$. The following holds:

\begin{enumerate}

\item If $v\in M$, then  $|\{N \in \mathcal{N}: soc(N)=S\}| = 4^{\dim(soc(M)) -1}$.

\item If $M-soc(M)\neq \emptyset$, then $|\mathcal{N}|= \dfrac{4^{\dim (M)}-4^{\dim(soc(M))}}{12}$.

\end{enumerate}
\end{lemma}

\begin{proof}

\begin{enumerate}
 

\item $\mathcal{N}_{S}=\{N \leq \mathcal{N}: soc(N)=S \}$. Note that $N \in \mathcal{N}_{S}$ if and only if there exists $x \in \left( T^{-1}(S) -  K\right)\cap M$ such that $N=\langle x, r(x) \rangle_F$ (by Lemma \ref{lem2.3}, part $3$). Suppose $v\in M$, then $v\in \left( T^{-1}(S) -  K\right)\cap M$ and $\langle v, r(v) \rangle_F\in \mathcal{N}_{S}$, so that $\mathcal{N}_{S}\neq \emptyset$. On the other hand, since $ T^{-1}(S)=  \sqcup_{ \lambda \in F} ( \lambda v+K ) $ and  $(\lambda v+K)\cap M=\lambda v+(K\cap M)$ for all $\lambda \in F$ (because $v\in M$), then $\left(T^{-1}(S) -  K\right)\cap M=\sqcup_{ \lambda \in F-\{0\}}  \lambda v+ (K\cap M)=\sqcup_{ \lambda \in F-\{0\}}  \lambda v+ soc(M)$ where the last equality is by Lemma \ref{lem2.3} (part $1$). Hence there are $|\left(T^{-1}(S) -  K\right)\cap M|=|\sqcup_{ \lambda \in F-\{0\}}  \lambda v+ soc(M)|=\sum_{ \lambda \in F-\{0\}}  |\lambda v+ soc(M)|=3 |soc(M)|$ elements $x\in M$  such that $\langle x , r(x) \rangle_F\in \mathcal{N}_{S}$, some of which generate the same element of $\mathcal{N}_{S}$. In fact, if $N \in \mathcal{N}_{S}$, each one of the $12$ elements $y\in N-S$ satisfies that $ \langle y , r(y) \rangle_F =N$. This implies that $\mathcal{N}_{S}$ has size $3|soc(M)|/12=\left(3 \cdot 4^{\dim (soc(M))}\right)/12=4^{\dim(soc(M)) -1}$.

\item Suppose $M-soc(M)\neq \emptyset$, and $W:=\{ N-soc(N): N\in \mathcal{N} \}$. It will be shown that $W$ is a partition of $M- soc(M)$. Let $W_{i}\in W$ be such that $W_{i} = N_{i}- soc(N_{i})$ with $N_{i}\in \mathcal{N}$ for $i=1,2$. Suppose $W_{1}\cap W_{2} \neq \emptyset$ and let $u \in W_{1} \cap W_{2}$. Since $N_{1}=\left\langle u, r(u) \right\rangle_{F} = N_{2}$ (by Lemma \ref{lem2.3}, part $3$),  $W_1=W_2$. On the other hand, if $m \in M- soc(M)$, then  $N:=\langle m, r(m) \rangle_F \in \mathcal{N}$ (by Lemma \ref{lem2.3}, part $3$) and $m \in N- soc(N)$. Therefore  $|M- soc(M)|=\sum_{Y\in W} |Y|=|W||Y|=|W|\cdot 12 $ and so $|W|=\left( 4^{\dim(M)}-4^{\dim( soc(M))}\right)/12$.
\end{enumerate}
\end{proof}

 If $q$ is a power of a prime number, $[y]_{q}:=(q^{y}-1)/(q-1)$ for all $y\in \mathbb{Z}_{>0}$. The Gaussian binomial coefficient (or $q$-binomial coefficient) $\binom{m}{t}_{q}$ is the number subspaces of $\mathbb{F}_{q}^{m}$ of dimension $t$ (see \cite{goldman, gasper}). This number is described by the formula $\binom{m}{t}_{q}=\dfrac{[m]_{q}!}{[t]_{q}!\cdot[m-t]_{q}!}$   where $ [x]_{q}! = [1]_{q}\cdot[2]_{q}\cdots[x]_{q} $ if  $x\neq 0$, and $ [x]_{q}! = 1$ if $x=0$.\\
If $W, H\leq V$, and $s$ and $t$ are non-negative integers, then 
$\mathcal{L}_{t,s}(W):=\{M\leq W : \, M\cong_{R} tR\oplus sF \}$, $\mathcal{L}'_{t,s}(W):=\{M\leq \mathcal{L}_{t,s}(W) : \mathbf{1}\subset M \}$, and $\mathcal{U}_{s}(W, H):=\{L\leq H: W\cap L=0 \wedge  L \cong sF\}$.

\begin{theorem}\label{thm3.4}
Let $s$ and $t$ be non-negative integers, and $0\neq N,M \leq V$.  The following holds:
\begin{enumerate}
\item $|\mathcal{L}_{0, s}(M)|=  \binom{\dim(soc(M))}{s}_{4}$. 

\item If $soc(N)\subsetneq soc(M)$ and $1\leq k \leq \dim(soc(M)/soc(N)) $, $\mathcal{U}_{k}(N,$ $ M)\neq \emptyset$ and $
 |\mathcal{U}_{k}(N, M)|=\prod_{i=0}^{k-1} \dfrac{4^{\dim(soc(M))}- 4^{\dim(soc(N))+i}}{4^{\dim(soc(N))+k}- 4^{\dim(soc(N))+i}}
$.

\item If $\mathcal{L}_{t,0}(M)\neq \emptyset$, the mapping $\alpha: \mathcal{L}_{t,0}(M)\to  \mathcal{L}_{0, t}(I\cap M)$ given  by $\alpha(Y):=\soc(Y)$ is surjective and   $|\alpha^{-1}(X)|=4^{t\dim(soc(M))-t^2}$ for all $X\in   \mathcal{L}_{0, t}(I\cap M)$. In particular, $|\mathcal{L}_{t,0}(M)|= \binom{\dim(I\cap M)}{t}_{4} \cdot  4^{t\dim(soc(M))-t^2}$.

\item  If  $1 \leq t\leq  \lfloor n/2 \rfloor$, $1\leq s \leq  \lceil n/2 \rceil - t$, $N\in \mathcal{L}_{t,0}(V)$ and $L \in \mathcal{U}_{s}(N, V)$,  $|\mathcal{L}_{t,s}(V)|=|\mathcal{L}_{t,0}(V)|(|\mathcal{U}_{s}(N, V)|/|\mathcal{U}_{s}(N, N\oplus L)|)$.

\item If $s\geq 2$ and $W\in \mathcal{U}_{s-1}(\mathbf{1}, V)$,  $|\mathcal{L}'_{0,s}(V)|=|\mathcal{U}_{s-1}(\mathbf{1}, V)|/$ $|\mathcal{U}_{s-1}(\mathbf{1},$ $ \mathbf{1}\oplus W)|$.

\item If  $\mathcal{L}'_{t,0}(V)\neq \emptyset$ and $W\in \mathcal{U}_{t-1}(\mathbf{1}, I)$,
$|\mathcal{L}'_{t,0}(V)|:= (|\mathcal{U}_{t-1}(\mathbf{1}, I)|/$ $|\mathcal{U}_{t-1}( $ $\mathbf{1},\mathbf{1}\oplus W)|) \cdot  4^{t\lceil n/2 \rceil - t^{2} } $. 
 
\item If  $1 \leq t\leq  \lfloor n/2 \rfloor$, $1\leq s \leq  \lceil n/2 \rceil - t$, $N\in\mathcal{L}_{t,0}(V)$, then $|\mathcal{L}'_{t,s}(V)|= |\mathcal{L}'_{t,0}(V)|\cdot  (|\mathcal{U}_{s}$ $(N, V)|/ |\mathcal{U}_{s}(N, N\oplus J_{0})|)+ |\mathcal{L}_{t,0}(V)-\mathcal{L}'_{t,0}(V)|(|\mathcal{U}_{s-1}(N\oplus \mathbf{1}, V)|/ |\mathcal{U}_{s-1}(N\oplus \mathbf{1},N\oplus \mathbf{1} \oplus J_{1} )|)$ where $J_{0}\in \mathcal{U}_{s}(N, V)$ and $J_{1}\in \mathcal{U}_{s-1}(N\oplus \mathbf{1}, V)$, if $n$ is even; and  $|\mathcal{L}'_{t,s}(V)|=|\mathcal{L}_{t,0}(V)| (|\mathcal{U}_{s-1}(N$ $\oplus  \mathbf{1}, V)|/$ $|\mathcal{U}_{s-1}(N\oplus  \mathbf{1},N\oplus  \mathbf{1}\oplus J_{2})|)$ where $J_{2}\in \mathcal{U}_{s-1}(N\oplus \mathbf{1},V)$, if $n$ is odd.

\end{enumerate}

\end{theorem}

\begin{proof}

\begin{enumerate}

 \item  It is clear. 

\item Let $soc(N)\subsetneq soc(M)$. It is clear that if $k=1$, the statement holds. Let $2\leq k \leq \dim(soc(M)/$ $soc(N))$ and $b_{i}:=\binom{\dim(soc(N))+i}{1}_{4}$ for $i=1,...,k-1$.  Note that (up to reordering) $\{U_{i}\}_{i=1}^{k}$  is an  independent set  of simple submodules of $M$ such that  $\oplus_{i=1}^{k} U_{i}\in \mathcal{U}_{k}(N, M)$ if and only if $U_{1}\nsubseteq N$ and $U_{i}\nsubseteq N \oplus( \oplus_{j=1}^{i-1}U_{j})$ for all $i=2,...,k$. To chose $U_{1}\leq M$ simple such that $N+U_{1}=N\oplus U_{1}$ there are $\binom{ \dim(soc(M)) }{1}_{4} - \binom{\dim(soc(N))}{1}_{4}$ choices; to chose $U_{2}\leq M$ simple such that $(N\oplus U_{1})+U_{2}=N\oplus (U_{1}\oplus U_{2})$ there are $\left[\binom{\dim(soc(M))}{1}_{4}- b_{1}\right]$ choices; to chose $U_{3}\leq M$ simple such that $(N\oplus (U_{1}\oplus U_{2}))+U_{3}=N\oplus (U_{1}\oplus U_{2} \oplus U_{3})$ there are $\left[\binom{\dim(soc(M))}{1}_{4}- b_{2}\right]$ choices. Thus, in general, to construct an independent set $\{U_{i}\}_{i=1}^{k}$ of simple submodules (of $M$) such that $N+ (\sum_{i=1}^{k} U_{i})=N\oplus (\oplus_{i=1}^{n} U_{i})$, there are
$\frac{1}{k!}\cdot \prod_{i=0}^{k-1}\left[\binom{\dim(soc(M))}{1}_{4}- b_{i} \right] =    \frac{1}{k!}\cdot \prod_{i=0}^{k-1} \frac{4^{\dim(soc(M))}- 4^{\dim(soc(N))+i}}{3}$
possibilities, but some of these sets produce the same direct sum with $N$. Let $\{Z_{i}\}_{i=1}^{k}$  be an independent collection of simple submodules of $M$ such that $\oplus_{i=1}^{k} Z_{i}\in \mathcal{U}_{k}(N, M)$. Note that (up to reordering) $\{J_{i}\}_{i=1}^{k}$ is an independent collection of simple submodules of $W:=N\oplus (\oplus_{i=1}^{k} Z_{i})$ such that $W=N\oplus (\oplus_{i=1}^{k} J_{i})$ if and only if $J_{1}\nsubseteq N$ and $J_{i}\nsubseteq N \oplus( \oplus_{l=1}^{i-1}J_{l})$ for all $l=2,...,k$. Hence, to construct an independent set $\{J_{i}\}_{i=1}^{k}$ of simple submodules (of $W$) such that $\oplus_{i=1}^{k} J_{i}\in \mathcal{U}_{k}(N, W)$ there are  
$
 \frac{1}{k!}\cdot \prod_{i=0}^{k-1} \left[\binom{\dim(soc(W))}{1}_{4}- b_{i}\right]        =\frac{1}{k!}\cdot \prod_{i=0}^{k-1} \frac{4^{\dim(soc(N))+k}- 4^{\dim(soc(N))+i}}{3}                                                               
$
possibilities. Hence $|\mathcal{U}_{k}(N,M)|=\prod_{i=0}^{k-1} \frac{4^{\dim(soc(M))}- 4^{\dim(soc(N))+i}}{4^{\dim(soc(N))+k}- 4^{\dim(soc(N))+i}}.
$

 \item By Lemma \ref{lem2.3} (part $4$), $\alpha$ is surjective with codomain $\mathcal{L}_{0, t}(I\cap M)$. Thus, the equivalence relation $\backsim$ in $\mathcal{L}_{t,0}(M)$ defined by $U \backsim U'$ if $\alpha(U)=\alpha(U')$ is such that $|\mathcal{L}_{t,0}(M)/\backsim |=|\mathcal{L}_{0, t}(I\cap M)|$. On the other hand, let $X=S_1 \oplus \cdots \oplus S_t \in \mathcal{L}_{0, t}(I\cap M)$. Then each $S_{i}$ extents to $4^{\dim(soc(M)) -1}$ modules in $\mathcal{L}_{1,0}(M)$ (by Lemma \ref{lem3.3}, part $1$), and so there are $(4^{\dim(soc(M)) -1})^t$ subsets of size $t$ of $\mathcal{L}_{1,0}(M)$, formed by socle-extensions of the $S_{i}$'s, whose sum is a socle-extension of $X$. Some of these produce the same socle-extension of $X$.  In fact, if $T \in \mathcal{L}_{t,0}(M)$ is such that $T = T_1 \oplus \cdots \oplus T_t $ with $soc(T_{i})=S_{i}$ for all $i$, $T$ can be written as a sum of socle-extensions of the $S_{i}$'s as many times as $(4^{t-1})^{t}$, which is the number of sets $\{Y_i\}_{i=1}^{t}\subset \mathcal{L}_{1,0}(T)$   with $S_{i}=soc(Y_i)$, by Lemma \ref{lem3.3} (part $1$). Hence, if $ \overline{U} \in \mathcal{L}_{t,0}(M)/\backsim$, $|\overline{U}|=|\alpha^{-1}(soc(U))|$ $=(4^{\dim((soc(M))-1})^{t}/(4^{t-1})^{t}$. Therefore,
$ |\mathcal{L}_{t,0}(M)|= |\mathcal{L}_{0, t}(I\cap M)| \ |\overline{U}|= \binom{\dim(I\cap M)}{t}_4 \cdot $ $ 4^{t \dim((soc(M))-t^{2}}$.

\item Let $N\in \mathcal{L}_{t,0}(V)$ and $L \in \mathcal{U}_{s}(N, V)$. Since $N\oplus L\in \mathcal{L}_{t,s}(V)$ and there exist $|\mathcal{U}_{s}(N, N\oplus L)|$ modules in $\mathcal{U}_{s}(N, V)$ whose sum with $N$ is $N\oplus L$,  then $|\mathcal{L}_{t,s}(V)|=|\mathcal{L}_{t,0}(V)|(|\mathcal{U}_{s}(N, V)|/$ $|\mathcal{U}_{s}(N, N\oplus L)|)$.

\item Let $\mathcal{L}'_{0,s}(V)=\{M\in \mathcal{L}_{0,s}(V): \mathbf{1}\subset M\}$. Then $\mathcal{L}'_{0,s}(V)\neq \emptyset$. Note that $M\in \mathcal{L}'_{0,s}(V)$ if and only if there exists $W\in \mathcal{U}_{s-1}(\mathbf{1}, V)$ such that $M=\mathbf{1}\oplus W$. Thus, if $W\in \mathcal{U}_{s-1}(\mathbf{1}, V)$, then  $|\mathcal{L}'_{0,s}(V)|=|\mathcal{U}_{s-1}(\mathbf{1}, V)|/|\mathcal{U}_{s-1}(\mathbf{1},$ $\mathbf{1}\oplus W)|$ (by a similar argument to the one in part $4$).  

\item Let $\mathcal{L}'_{t,0}(V)\neq \emptyset$.  Note that  $M\in \mathcal{L}'_{t,0}(V)$ if and only if $M\in \mathcal{L}_{t,0}(V)$ and   $soc(M)=\mathbf{1}\oplus W \subseteq I$ for some $W\in \mathcal{L}_{0,t-1}(I)$ (by Lemma \ref{lem2.3}, part $4$), this is only possible if $n$ is even. Let  $W\in \mathcal{U}_{t-1}(\mathbf{1}, I)$. Then there are $|\mathcal{U}_{t-1}(\mathbf{1}, I)|/|\mathcal{U}_{t-1}(\mathbf{1}, \mathbf{1}\oplus W)|$ modules in $\mathcal{L}_{0, t}(I)$ containing $\mathbf{1}$, each one with $4^{t\lceil n/2 \rceil-t^2}$ socle-extensions in $\mathcal{L}_{t,0 }(V)$ (by part $3$), and so
$|\mathcal{L}'_{t,0}(V)|=(|\mathcal{U}_{t-1}(\mathbf{1}, I)|/|\mathcal{U}_{t-1}(\mathbf{1}, $ $\mathbf{1}\oplus W)|) 4^{t\lceil n/2 \rceil-t^2}.$

\item  Let  $1 \leq t\leq  \lfloor n/2 \rfloor$, $1\leq s \leq  \lceil n/2 \rceil - t$, and $\mathcal{L}'_{t,s}(V)=\{M\in \mathcal{L}_{t,s}(V): \mathbf{1}\subset M\}$. Then $\mathcal{L}_{1}:=\{M\in \mathcal{L}'_{t,s}(V): \mathbf{1}\subseteq I\cap M \}\neq \emptyset$ if and only if $n$ is even; $\mathcal{L}_{2}:=  \{M\in \mathcal{L}'_{t,s}(V): \mathbf{1}\nsubseteq I\cap M \}\neq \emptyset$ if and only if $n$ is odd; and $\mathcal{L}'_{t,s}(V)= \mathcal{L}_{1} \sqcup \mathcal{L}_{2}$. Let $n$ be even, then $\mathbf{1}\subseteq I$ and so $A:=\{M\in \mathcal{L}_{1}:\mathcal{L}'_{t,0}(M)\neq \emptyset\}\neq \emptyset$. 
 If $N_{0}\in \mathcal{L}'_{t,0}(V)$ and  $J_{0}\in \mathcal{U}_{s}(N_{0}, V)$, then $|A|= |\mathcal{L}'_{t,0}(V)|\cdot  (|\mathcal{U}_{s}(N_{0}, V)|/ |\mathcal{U}_{s}(N_{0}, N_{0}\oplus J_{0})|)$. Note that $\mathcal{L}'_{t,0}(V)\subsetneq \mathcal{L}_{t,0}(V)$ if and only if $t<n/2 $. If  $A\neq \mathcal{L}_{1}$ (i.e. $t<n/2$), $N_{1}\in \mathcal{L}_{t,0}(V)-\mathcal{L}'_{t,0}(V)$ and $J_{1}\in \mathcal{U}_{s-1}(N_{1}\oplus \mathbf{1}, V)$, then $ |\mathcal{L}_{1}-A|=|\mathcal{L}_{t,0}(V)-\mathcal{L}'_{t,0}(V)|(|\mathcal{U}_{s-1}(N_{1}\oplus \mathbf{1}, V)|/ |\mathcal{U}_{s-1}(N_{1}\oplus \mathbf{1},N_{1}\oplus \mathbf{1} \oplus J_{1} )|)$.  Let $n$ be odd and $N\in \mathcal{L}_{t,0}(V)$, then $\mathbf{1}\nsubseteq soc(N)\subseteq I$. Thus, if $J_{2}\in \mathcal{U}_{s-1}(N\oplus \mathbf{1},V)$, then  $|\mathcal{L}_{2}|=|\mathcal{L}'_{t,s}(V)|=|\mathcal{L}_{t,0}(V)| (|\mathcal{U}_{s-1}(N\oplus  \mathbf{1}, V)|/|\mathcal{U}_{s-1}(N\oplus  \mathbf{1},N\oplus  \mathbf{1}\oplus J_{2})|)$.

\end{enumerate}

\end{proof}
 
In \cite{abualrub}, the authors computed all the linear cyclic codes that are reversible-complementary\footnote{named there as ``reversible  complement cyclic codes'' \cite[Definition 3]{abualrub}} of lengths $7$, $9$, $11$ and $13$. They only obtained the repetition code for $n=7,11,13$ and two codes distinct from the repetition code for $n=9$. By using Corollary \ref{cor-const-tRsF} (part $2$), one can produce $15$ distinct non-isomorphic modules containing the repetition code for $n=9$, and for each isomorphism type, the number of codes containing the repetition code can be computed using Theorem \ref{thm3.4}. Thus, the method presented here produces many more reversible-complementary codes than the one in \cite{abualrub}. On the other hand, \cite{fripertinger} introduced formulas for counting the invariant subspaces under a linear endomorphism of $\mathbb{F}_{q}^{n}$. Although this approach is more general than the one offered here for reversible codes, it is not suited for counting reversible-complementary codes.

\section*{Conclusion}

Some upper bounds for the minimum Hamming distance of reversible codes are given. In some cases, these are tighter than the Singleton bound. A method to construct reversible and reversible-complementary linear codes over $\mathbb{F}_{4}$ is presented. This method complements other known constructions, offering better results in terms of the number of codes produced. Explicit formulas for the number of the studied codes are presented, solving the problem of computing all the reversible codes containing the repetition code.

\end{document}